\documentclass[reqno]{amsart}
\usepackage{bbm}
\usepackage{amsfonts}
\usepackage{mathrsfs}
\usepackage{hyperref}
\usepackage{amssymb}
\usepackage{CJK}
 \newtheorem{thm}{Theorem}[section]
 \newtheorem{cor}[thm]{Corollary}
 \newtheorem{prop}[thm]{Proposition}
 \newtheorem{lem}[thm]{Lemma}

 \theoremstyle{remark}
 \newtheorem{rem}[thm]{Remark}
 \numberwithin{equation}{section}
\DeclareMathOperator{\sign}{sign}
\DeclareMathOperator{\order}{ord}

  \newcommand{\f}{\mathbb{F}_{q}}
 \newcommand{\np}{\textbf{NP}}
  \newcommand{\rp}{\textbf{RP}}

\begin{document}
\title[On min. distances of ECAG codes]
  { On the minimum distance of elliptic curve codes}

\author{Jiyou Li}
\address{Department of Mathematics, Shanghai Jiao Tong University, Shanghai, P.R.
China and Department of Mathematics, Unversity of Delaware, USA}

\email{lijiyou@sjtu.edu.cn}

\author{Daqing Wan}
\address{Department of Mathematics, University of California, Irvine, CA 92697-3875, USA}
\email{dwan@math.uci.edu}

\author{Jun Zhang }
\address{School of Mathematics, Capital Normal University}
\email{junz@cnu.edu.cn}


\thanks{The work of Jiyou Li is supported by the National Science Foundation of China
(11001170) and Ky and Yu-Fen Fan Fund Travel Grant from the AMS. The
research of Daqing Wan is partially  supported by NSF. This research
of Jun Zhang is supported by the National Key Basic Research Program
of China (2013CB834204), the National Natural Science Foundation of
China (61171082, 10990011 and 60872025).}


\begin{abstract}
Computing the minimum distance of a linear code is one
 of the fundamental problems in algorithmic coding theory. Vardy~\cite{var97} showed that it is an \np-hard problem for general linear codes. In practice, one often uses codes with additional mathematical structure, such as AG codes.
 For AG codes of genus $0$ (generalized Reed-Solomon codes), the minimum distance has a simple explicit formula.
 An interesting result of Cheng~\cite{chengqi}
 says that the minimum distance problem is already \np-hard (under \rp-reduction)
 for general elliptic curve codes (ECAG codes, or AG codes of genus $1$).  In this paper, we show that the minimum
 distance of ECAG codes also has a simple explicit formula if the evaluation set is suitably large (at least
 $2/3$ of the group order).
 Our method is purely combinatorial and based on a new sieving technique from  the first two authors \cite{LW1}. This method
 also proves a significantly stronger version of the MDS (maximum distance separable) conjecture for ECAG codes.



\end{abstract}

\maketitle \numberwithin{equation}{section}
\newtheorem{theorem}{Theorem}[section]
\newtheorem{lemma}[theorem]{Lemma}
\newtheorem{example}[theorem]{Example}
\allowdisplaybreaks

\section{Introduction}
Let $\f^n$ be the $n$-dimensional vector space over the finite field $\f$ with $q$ elements. For any vector $ {x}=(x_1,x_2,\cdots,x_n)\in \f^n$, the \emph{Hamming weight} $\mathrm{Wt}( {x})$ of $ {x}$ is defined to be the number of non-zero coordinates, i.e.,
$$\mathrm{Wt}( {x})=\#\left\{i\,|\,1\leqslant i\leqslant n,\,x_i\neq 0\right\}\ .$$
A \emph{linear $[n,k]$ code} $C$ is a $k$-dimensional linear subspace of $\f^n$. The \emph{minimum distance} $d(C)$ of $C$ is the minimum Hamming weight of all non-zero vectors in $C$, i.e.,
$$d(C)=\min\{\mathrm{Wt}( {c})\,|\, {c}\in C\setminus\{ {0}\}\}\ .$$
A linear $[n,k]$ code $C\subseteq \f^n$ is called a $[n,k,d]$ linear
code if $C$ has minimum distance $d$. A well-known trade-off between
the parameters of a linear $[n,k,d]$ code is the Singleton bound
which states that
$$d\leqslant n-k+1\ .$$
 An $[n,k,d]$ code is called a \emph{maximum distance separable}
  (MDS) code if $d=n-k+1$.
 The \emph{dual code} $C^\bot$ of $C$ is defined as the set
\[
  \left\{ {x}\in \f^n\,|\, {x}\cdot {c}=0\,\mbox{for all } {c}\in C\right\},
\]
where $ {x}\cdot {c}$ is the inner product of vectors $ {x}$ and $ {c}$, i.e.,
\[
     {x}\cdot {c}=x_1c_1+x_2c_2+\cdots+x_nc_n\ .
\]

Computing the minimum distance of a linear code is one of the most
important problems in algorithmic coding theory. It was proved to be \np-hard for
general linear codes in \cite{var97}. The gap version of the problem
was also shown to be \np-hard in \cite{DMS}. And the same paper showed
 that approximating the minimum distance of a linear code cannot
  be achieved in randomized polynomial time to the factor
  $2^{\log^{1-\epsilon}n}$ unless $\mathrm{NP}\subseteq \mathrm{RTIME}
  (2^{\mathrm{polylog}(n)})$. In \cite{CW09}, Cheng and the
  second author derandomized the reduction and showed there is
   no deterministic polynomial time algorithm to
   approximate the minimum distance to any constant factor
    unless $\mathrm{NP}=\mathrm{P}$. And they proved that approximating the minimum distance of a linear code cannot be achieved in deterministic polynomial time to the factor $2^{\log^{1-\epsilon}n}$ unless $\mathrm{NP}\subseteq \mathrm{RTIME}(2^{\mathrm{polylog}(n)})$.

Despite the above complexity results, it is more interesting to compute the minimum distance of linear codes that are used in practical applications.
An important class of such codes is algebraic geometry (AG) codes with parameters $[n,k,d]$ as defined in Section~4.
The minimum distance of such AG codes from algebraic curves of genus $g$ is known to satisfy the inequality
$$n-k-g+1\leq d\leq n-k+1.$$
In the simplest case $g=0$, i.e., generalized Reed-Solomon codes, the minimum distance has the simple formula $d=n-k+1$.
In the next simplest case $g=1$,
either $d=n-k$ or $d=n-k+1$, and Cheng~\cite{chengqi} showed that determining the minimum distance of ECAG codes
between the two options is \np-hard under \rp-reduction. For genus $g\geq 2$, there is no such complexity result so far.
But it is believed to be an \np-hard problem as well.

We are interested in positive results for determining the minimum distance of ECAG codes.
It was shown in \cite{chengqi}, and also in \cite{ZFW} from a different aspect, that
computing the minimum distance of an ECAG code is equivalent to a
subset sum problem (SSP) in the group of rational points on the elliptic curve.
We now make this more precise.

Let $E$ be an elliptic curve over the finite field $\f$. Let $G$ be
the group of $\f$-rational points on the elliptic curve $E$. The Hasse bound shows that
$||G|-(q+1)| \leq 2\sqrt{q}$.
Let $D\subseteq G$ be a nonempty subset of cardinality $n$, which will be our evaluation set for ECAG code.
For a positive integer $1\leq k\leq n< |G|$ and element $b\in G$, let $N(k, b,
D)$ be the number of $k$-subsets $T\subseteq D$ such that
$\sum_{x\in T}x=b$. The counting version of the $k$-subset sum problem for the pair $(G, D)$ is to compute $N(k,b, D)$.
The minimum distance of the ECAG $[n,k]$-code is equal to $n-k$ if and only if the number $N(k,b,D)$
is positive. This $k$-subset sum problem is in general \np-hard if the evaluation set $D$ is small.  On the other hand,
the dynamic programming method implies that there is a polynomial time algorithm to compute
$N(k,b,D)$ if $n=|D|$ is large, say, $n=|G|^{\delta}$ for some constant $\delta>0$.

In this paper, we obtain an asymptotic formula for $N(k,b,D)$ if $n=|D|$ is suitably large, say,
$|D| > (\frac{2}{3} +\epsilon)  |G|$.
As an application, we show that if the cardinality $n$ of the evaluation set is
suitably large (at least $2/3$ of the group order), then the minimum distance of an ECAG code $[n,k]$ is
always $n-k$. We conjecture that the condition $|D| > (\frac{2}{3} +\epsilon)  |G|$ in our
results can be improved to $|D| > (\frac{1}{2} +\epsilon)  |G|$.
Our main technical tool is the sieve method of the first two authors \cite{LW2}.

To describe the asymptotic formula, we introduce more notations.
Let
$\widehat{G}$ be the group of additive characters of $G$.  Note that
$\widehat{G}$ is isomorphic to $G$. Define
$$\Phi(D)=\max_{\chi\in
\widehat{G}, \chi\ne \chi_0}|\sum_{a\in D}\chi{(a)}|.$$

Our main technical result is the following asymptotic formula for $N(k,b,D)$.

\begin{thm}\label{thm1}
Notations as above. We have
\begin{align*}
\left | N(k, b, D)-{|G|^{-1}} {n \choose k}\right| &\leq \frac {|S|}
{|G|}{\Phi(D)+k-1 \choose k}+\frac 1 {|G|}{\frac {n+\Phi(D)}
2\choose k}\nonumber \\ &+\frac 1 {|G|}\sum_{2<d\leq k \atop d\mid
\exp(G)}
 \phi(d){\frac {n+\Phi(D)}{d}+k-1 \choose k},
 \end{align*}
 where $S$ is the set of characters in $\widehat{G}$ which have order greater
 than $k$ and $\exp(G)$ is the exponent of $G$.

\end{thm}

We apply this theorem to determine the minimum distance of ECAG
codes (for details see Section~\ref{sec4}) and obtain
\begin{thm}\label{thm1.3}
Suppose that $n\geq (\frac 23+\epsilon)q$ and $q>\frac 4 {\epsilon^2}$,
where $\epsilon $ is positive. There is a positive constant
$C_{\epsilon}$ such that if
$C_{\epsilon}\ln{q}<k<n-C_{\epsilon}\ln{q}$, then ECAG codes $[n,k]$
have the deterministic minimum distance $n-k$.
\end{thm}

If we allow the length of the codes to be larger, we then have a
better bound on $k$.

\begin{thm}\label{thm1.2}
If $n\ge q+2$, then for $q>64$ and
 $3<k<q-1$, then ECAG $[n,k]$ codes have the deterministic minimum distance $n-k$.
\end{thm}


Since one can check the cases $q\leq 64$ by a computer search,
we have a complete result for the minimum distance of the ECAG code $[n,k]$
if $n\geq q+2$. This gives a
new proof of MDS conjecture on ECAG codes, in a purely combinatoric
method.
We now explain this application and its improvement.

Recall that an  $[n,k,d]$ code is called a \emph{maximum distance
separable}
  (MDS) code if $d=n-k+1$.
MDS codes have a lot of advantages~\cite{Mac}.
However, MDS
codes are very rare, and so far, not too many MDS codes have been
found.
The Main Conjecture on MDS Codes states that for every linear $[n, k]$ MDS code over $\f$, if $1 < k < q$,
then $n\leq q+1$, except when $q$ is even and $k=3$ or $k=q-1$, in which cases $n\leq q + 2$.

The most well-known MDS codes are Reed-Solomon codes. Since the
evaluation set of a Reed-Solomon code can not exceed the finite
field, the MDS conjecture always holds in this case. The MDS conjecture
was proved whenever $q\leq 11$ or $k\leq 5$ by using the theory of
finite geometries. Since the most popular candidates for MDS codes
are the Goppa codes constructed from algebraic curves of small
genus and algebraic geometry (AG) codes, people turned to concentrate
the MDS conjecture for AG codes. As AG codes have algebraic and
geometric properties, there are a lot of new algebraic-geometric
methods to apply, while the general MDS property is more of a
combinatorial property. The MDS conjecture for ECAG codes
was first proved by Katsman and Tsfasman
in~\cite{KT}. Munucra \cite{Mun} translated the conjecture for AG
codes to another conjecture concerning the arithmetic of the curves.
He then proved it for codes arising from elliptic curves, and curves
of genus $2$ when $q > 83$. Walker~\cite{Wal} presented a new
approach to the problem in the case of elliptic curves by proving a
statement about the geometry of the curve after a certain embedding.

In the case of hyperelliptic curves, for fixed genus $g$,
Moer~\cite{Boer} showed that MDS conjecture holds when $q$ is big
enough. Chen \cite{Chen} proved that there is a constant $C(g)$
depending only on the genus $g$ such that the MDS conjecture is true
when $q>C(g)$. And later in~\cite{Chen1} Chen and Yau gave an upper
bound of $C(g)$ which not only affirmatively answered the question
asked by Munucra in~\cite{Mun}, but also improved the result
in~\cite{Boer} a lot.

As we pointed out  that the minimum distance of an ECAG code has
only two options and determining it is equivalent to an SSP problem,
the MDS conjecture on ECAG codes is naturally reduced to a subset
sum problem of the group of rational points on the elliptic curve.
By Theorem~\ref{thm1.2}, we have
\begin{thm}
For $q>64$, MDS conjecture for ECAG codes holds.
\end{thm}

By Theorem~\ref{thm1.3}, if some restriction on the dimension $k$ is allowed,
we can significantly improve the lower bound $q+2$.

\begin{thm}\label{thm1.5}
 Suppose that $n\geq (\frac 23+\epsilon)q$ and $q>\frac 4 {\epsilon^2}$,
where $\epsilon$ is
  positive. Then there is a positive constant $C_{\epsilon}$ such
that if
 $C_{\epsilon}\ln{q}<k<n-C_{\epsilon}\ln{q}$, then
 there is no MDS ECAG code with parameter $[n,k]$.
\end{thm}

For small $k$, one can directly check if the ECAG code is MDS or not. For large $k$, by the
duality, it can be reduced to the former case. From Theorem~\ref{thm1.5}, we shall see that to get a long MDS code for fixed alphabet size $q$, Reed-Solomon codes are always the best choices.

This paper is organized as follows. Section~2 recalls the sieve method of the first two authors. Section~3 uses the sieve method to get an estimate of counting subset sum problems on any large subset of the rational point group of an elliptic curve. And Section~4 describes the relation between minimum distance of ECAG codes and subset sum problems on the evaluation set of the ECAG code. The main theorems of this paper then follow.

\section{A distinct coordinate sieving formula}

In this section we introduce a sieving formula discovered by Li-Wan
\cite{LW1}. It significantly improves the classical
inclusion-exclusion sieve in several important cases. We recite it
here without proof. For details and related applications, we refer
to \cite{LW1,LW2}. Before we present the sieving formula, we
introduce some notations valid for the whole paper.

\begin{itemize}
  \item Let $D$ be an alphabet set, $X$ a finite set of vectors of length $k$ over $D$.
  \item Denote $\overline{X}=\{(x_1,x_2,\cdots,x_k)\in X\ | \ x_i\ne x_j, \forall i\ne j\}$ the pairwise distinct component subset.
  \item Let $S_k$ be the symmetric group on $\{1,2,\cdots, k\}$. For $\tau\in S_k$, the $\sign$ function is defined to be $\sign(\tau)=(-1)^{k-l(\tau)}$, where $l(\tau)$ is the number of
 cycles of $\tau$ including the trivial cycles which have length $1$.
  \item Let $\tau=(i_1i_2\cdots i_{a_1})(j_1j_2\cdots j_{a_2})\cdots(l_1l_2\cdots l_{a_s})$ with $1\leq a_i, 1 \leq i\leq s$ be any permutation, denote the $\tau$-symmetric subset
      \begin{equation}\label{1.1}
        \begin{array}{rl}
           X_{\tau}=&\left\{
(x_1,\dots,x_k)\in X\,|\,
 x_{i_1}=\cdots=x_{i_{a_1}},\cdots,\right. \\
   &\left.x_{l_1}=\cdots=x_{l_{a_s}}\right\}.
        \end{array}
      \end{equation}

  \item Let $f(x_1,x_2,\dots,x_k)$ be a complex valued function defined on $X$. Denote the distinct sum
$$F=\sum_{x \in \overline{X}}f(x_1,x_2,\dots,x_k), $$
and the $\tau$-symmetric sum
\[
F_{\tau}=\sum_{x \in X_{\tau} } f(x_1,x_2,\dots,x_k).
\]
\end{itemize}

We now present the sieving formula found in \cite{LW1}.

\begin{thm} \label{thm1.0}
Let $F$ and $F_{\tau}$ be defined as above. Then
\begin{align}
  \label{1.5} F=\sum_{\tau\in S_k}{\sign(\tau)F_{\tau}}.
    \end{align}
 \end{thm}

We notice that in this formula, there are at most $k!$ terms
(computable in many cases), which is significantly smaller than
$2^{{k\choose 2}}$, the needed number of terms by traditional
sieving approach.

%
%
%
 For $\tau\in S_k$, let $\overline{\tau}$ denote
 the conjugacy class determined by $\tau$ whose elements are
 permutations conjugate to $\tau$.
Conversely, in the case that we denote a conjugacy class by
$\overline{\tau}\in C_k$, $\tau$ is a correspondent representative
permutation. Since two permutations in $S_k$ are conjugate if and
only if they have the same type of cycle structure, $C_k$ is exactly
the set of all partitions of $k$.

The symmetric group of $k$ elements,  $S_k$, acts on $D^k$ naturally
by permuting coordinates. Given $\tau\in S_k$ and
$x=(x_1,x_2,\dots,x_k)\in D^k$,  $\tau\circ
x=(x_{\tau(1)},x_{\tau(2)},\dots,x_{\tau(k)}).$
  A subset $X$ in $D^k$ is defined to be \emph{symmetric} if for any $x\in X$ and
any $\tau\in S_k$, $\tau\circ x \in X $. In particular,  if  $X$ is
symmetric and $f$ is a symmetric function under the action of $S_k$,
we then get the following useful counting formula for (\ref{1.5}).
\begin{prop} \label{thm1.1} Let $C_k$ be the set of conjugacy  classes
 of $S_k$.  If $X$ is symmetric and $f$ is symmetric, then
 \begin{align}\label{7} F=\sum_{\tau \in C_k}\sign(\tau) C(\tau)F_{\tau},
  \end{align} where $C(\tau)$ is the number of permutations conjugate to
  $\tau$.
\end{prop}

For the purpose of our proof, we will also present several
combinatorial formulas.  A permutation $\tau\in S_k$ is said to be
of type $(c_1,c_2,\cdots,c_k)$ if $\tau$ has exactly $c_i$ cycles of
length $i$.  
Denote by $N(c_1,c_2,\dots,c_k)$ to be the number of
$k$-permutations of type $(c_1,c_2,\dots,c_k)$. It is well known
that
$$N(c_1,c_2,\dots,c_k)=\frac {k!} {1^{c_1}c_1! 2^{c_2}c_2!\cdots k^{c_k}c_k!}.$$

\begin{lem} \label{lem2.6}
If we define the generating function
\begin{align*}C_k(t_1,t_2,\dots,t_k)= \sum_{\sum
ic_i=k} N(c_1,c_2,\dots,c_k)t_1^{c_1}t_2^{c_2}\cdots t_k^{c_k},
 \end{align*}
and set $t_1=t_2=\cdots=t_k=q$, then
\begin{equation*}
    \begin{array}{rl}
      C_k(q,q,\dots,q)= & \sum_{\sum
ic_i=k} N(c_1,c_2,\dots,c_k)q^{c_1}q^{c_2}\cdots q^{c_k}\nonumber \\
      =& (q+k-1)_k
    \end{array}
\end{equation*}

 If we set $t_i=q$ for $d\mid i$ and $t_i=s$ for $d\nmid
i$, then
\begin{equation*}
    \begin{array}{rl}
&C_k(\overbrace{s,\cdots,s}^{d-1},q,\overbrace{s,\cdots,s}^{d-1},q,
\cdots)\\
 =& \sum_{\sum
ic_i=k} N(c_1,c_2,\cdots,c_k)q^{c_1}q^{c_2}\cdots s^{c_d}q^{c_{d+1}}\cdots\nonumber \\
=& k!\sum_{i=0}^{\lfloor k/d \rfloor}{\frac{q-s}{d}+i-1\choose
\frac{q-s}{d}-1} {s+k-di-1\choose s-1}\\
\leq& k!{s+k+(q-s)/d-1\choose k}.
    \end{array}
\end{equation*}
\end{lem}

\section{Subset Sum Problem in a Subset of the Rational Point Group}

\begin{lem}[Hasse-Weil Bound] \label{4.3}
Let $E$ be an elliptic curve over the finite field $\f$. Then the
number of rational points on $E$ has the following estimate
\[
  |\#E(\f)-q-1|\leq 2\sqrt{q}.
\]
\end{lem}

\begin{lem}[Structure of Rational Point Group]\label{4.2}
A group $G$ of order $N = q+1-m$ is isomorphic to $E(\f)$ for some
elliptic curve $E$ over $\f$ if and only if one of the following
conditions holds:

(i) $(q,m) = 1$, $|m|\leq 2\sqrt{q}$ and $G \cong \mathbb{Z}/A \times \mathbb{Z}/B$ where $B|(A,m-2)$.

(ii) $q$ is a square, $m = \pm 2\sqrt{q}$ and $G = (\mathbb{Z}/A)^2$ where $A =\sqrt{q}\mp1$.

(iii) $q$ is a square, $p\equiv 1\,\ (\mod 3)$, $m = \pm\sqrt{q}$ and $G$ is cyclic.

(iv) $q$ is not a square, $p = 2$ or $3$, $m = \pm\sqrt{pq}$ and $G$ is cyclic.

(v) $q$ is not a square, $p\equiv 3\,\ (\mod 4)$, $m = 0$ and $G$ is cyclic
or $q$ is a square, $p\equiv 1\,\ (\mod 4)$, $m = 0$ and $G$ is cyclic.

(vi) $q$ is not a square, $p\equiv 3\,\ (\mod4)$, $m = 0$ and $G$ is either cyclic or $G\cong \mathbb{Z}/M \times \mathbb{Z}/2$ where $M =\frac{q+1}{2}$ .

\end{lem}

According to Lemma \ref{4.2} on the structure of $E(\f)$, we may
suppose that $G=E(\f)\cong \mathbb{Z}/{n_1}\times \mathbb{Z}/{n_2}$ is a finite
abelian group. By Lemma \ref{4.3}, $G$ has order $q+1+c\sqrt{q}$,
with $|c|\leq 2$. Denote by $\exp(G)$ the exponent of $G$.
 Let $D\subseteq G$ be a nonempty subset of cardinality $n$. Let
$\widehat{G}$ be the group of additive characters of $G$.  Note that
$\widehat{G}$ is isomorphic to $G$. Define $s_{\chi}(D)=\sum_{a\in
D}\chi{(a)}$ and  $\Phi(D)=\max_{\chi\in \widehat{G}, \chi\ne
\chi_0}|s_{\chi}(D)|$. Let $N(k, b, D)$ be the number of $k$-subsets
$T\subseteq D$ such that $\sum_{x\in S}x=b$.
In the following theorem we will give an asymptotic bound for $N(k, b,
D)$ which ensures $N(k, b, D)>0$ when $G-D$ is not too large compared with $G$.

 \begin{thm}\label{lem1.1}Let $N(k, b, D)$ be defined as above.
\begin{align}\label{4.1}
\left | N(k, b, D)-{|G|^{-1}} {n \choose k}\right| &\leq \frac {|S|}
{|G|}{\Phi(D)+k-1 \choose k}+\frac 1 {|G|}{\frac {n+\Phi(D)}
2\choose k}\nonumber \\ &+\frac 1 {|G|}\sum_{2<d\leq k \atop d\mid
\exp(G)}
 \phi(d){\frac {n+\Phi(D)}{d}+k-1 \choose k},
 \end{align}
where $S$ is the set of characters which has order greater than $k$.
  \end{thm}

\begin{proof}
Let $X=D\times D \times \cdots \times D$ be the Cartesian product of
$k$ copies of $D$.
 Let $  \overline{X} =\left\{ (x_1,x_2,\dots,x_{k} )\in D^k \mid
 x_i\not=x_j,~ \forall i\ne j\} \right\}.$ It is clear that $|X|=n^k$ and
$|\overline{X}|=(n)_k$. We have

\begin{align*}
k!N(k, b, D)&={|G|^{-1}} \sum_{(x_1, x_2,\dots x_k) \in
\overline{X}}
\sum_{\chi\in \widehat{G}}\chi(x_1+x_2+\cdots +x_k-b)\\
&={|G|^{-1}} (n)_k+|G|^{-1} \sum_{\chi\ne \chi_0}\sum_{(x_1,
x_2,\cdots x_k) \in\overline{X}}\chi(x_1)\chi(x_2)\cdots \chi(x_k)\chi^{-1}(b)\\
&={|G|^{-1}}  {(n)_k}+{|G|^{-1}} \sum_{\chi\ne
\chi_0}\chi^{-1}(b)\sum_{(x_1,x_2,\dots x_k)
\in\overline{X}}\prod_{i=1}^{k} \chi(x_i).
\end{align*}
Denote $f_{\chi}(x)= f_{\chi}(x_1,x_2,\dots,x_{k})=
\prod_{i=1}^{k}\chi(x_i).$  For
  $\tau\in S_k$,  let
$$F_{\tau}(\chi)=\sum_{x\in X_{\tau}}f_{\chi}(x)=\sum_{x \in X_{\tau}}\prod_{i=1}^{k} \chi(x_i),$$
where $X_{\tau}$ is defined as in (\ref{1.1}). Obviously $X$ is
symmetric and $f_{\chi}(x_1,x_2,\dots,x_{k})$ is normal on $X$.
Applying (\ref{7}) in Corollary \ref{thm1.1}, we get
  \begin{align*}
 k!N(k, b, D)&={|G|^{-1}} {(n)_k}+{|G|^{-1}} \sum_{\chi\ne \chi_0}\chi^{-1}(b) \sum_{\tau\in
C_{k}}\sign(\tau)C(\tau) F_{\tau}(\chi),
  \end{align*}
 where $C_{k}$ is the set of conjugacy classes
 of $S_{k}$, $C(\tau)$ is the number of permutations conjugate to $\tau$.  If $\tau$
 is of type $(c_1, c_2, \dots, c_k)$, then
  \begin{align*}
F_{\tau}(\chi)&=\sum_{x \in X_{\tau}}\prod_{i=1}^{k} \chi(x_i)\\
&=\sum_{x \in X_{\tau}}\prod_{i=1}^{c_1} \chi(x_i)\prod_{i=1}^{c_2}
\chi^2(x_{c_1+2i})\cdots\prod_{i=1}^{c_k} \chi^k(x_{c_1+c_2+\cdots+k i})\\
 &=\prod_{i=1}^{k}(\sum_{a\in D}\chi^i(a))^{c_i}\\
 &= n^{\sum c_i m_i(\chi)} s_{\chi}(D)^{{\sum c_i (1-m_i(\chi))}},
\end{align*}
where $m_i(\chi)=1$ if $\chi^i=1$ and otherwise $m_i(\chi)=0$.

  Now suppose $\order(\chi)=d$ with $d\mid n_1n_2$. Note that
$C(\tau)=N(c_1,c_2,\dots,c_k)$.  In the case $d=2$, $s_\chi(D)$ is
an integer. Applying Lemma \ref{lem2.6}, we have
\begin{align*}
&\sum_{\tau\in C_{k}}\sign(\tau)C(\tau) F_{\tau}(\chi)\\
&=(-1)^k\sum_{\tau\in C_{k}}C(\tau) (-n)^{\sum c_i m_i(\chi)}
(-s_\chi(D))^{{\sum c_i (1-m_i(\chi))}}\\
&=(-1)^k k!\sum_{i=0}^{\lfloor k/2
\rfloor}{\frac{-n+s_\chi(D)}{2}+i-1\choose
i} {-s_\chi(D)+k-2i-1 \choose k-2i}\\
&=k!\sum_{i=0}^{\lfloor k/2 \rfloor}{\frac{n-s_\chi(D)}{2}\choose
i} {s_\chi(D)\choose k-2i}\\
&\leq k!{\frac {n+\Phi(D)} 2\choose k}.
\end{align*}
 The last inequality in the case $s_\chi(D)>0$ is from the identity
$$\sum_{i=0}^k{a \choose i}{b\choose k-i}={a+b\choose k}.$$
In the
case $s_\chi(D)<0$, since the summation has alternative signs, the
inequality follows from a simple combinatorial argument.

In the case $3 \leq d\leq k$, since $|s_\chi(D)|\leq \Phi(D)$, we
have
\begin{align*}
&\sum_{\tau\in C_{k}}\sign(\tau)C(\tau) F_{\tau}(\chi)\\
&\leq\sum_{\tau\in C_{k}}C(\tau) n^{\sum c_i m_i(\chi)}
\Phi(D)^{{\sum c_i (1-m_i(\chi))}}\\
&\leq k!{\frac {n+\Phi(D)}{d}+k-1 \choose k}.
\end{align*}

Similarly, if $\order(\chi)$ is greater than $k$,  then
\begin{align*}
\sum_{\tau\in C_{k}}\sign(\tau)C(\tau) F_{\tau}(\chi)\leq k!{
\Phi(D)+k-1 \choose k}.
\end{align*}

 Let $S$ be the set of characters which have order greater than $k$.
 Summing over all nontrivial characters,  we obtain
\begin{align*}
\left | N(k, b, D)-{|G|^{-1}} {n \choose k}\right| &\leq \frac {|S|}
{|G|}{\Phi(D)+k-1 \choose k}+\frac 1 {|G|}{\frac {n+\Phi(D)}
2\choose k}\nonumber \\ &+\frac 1 {|G|}\sum_{2<d\leq k \atop d\mid
\exp(G)}
 \phi(d){\frac {n+\Phi(D)}{d}+k-1 \choose k},
 \end{align*}
 where $\phi(d)$ is the number of characters in $\widehat{G}$ of order $d$.
 This completes the proof.
\end{proof}

\begin{cor}We have
$$\left| N(k, b, D)-{|G|}^{-1}{n \choose k}\right|\leq {M \choose k},
$$
where $M$ is defined as $$M=\max\left \{ {\Phi(D)+k-1\choose
k},{\frac {n+\Phi(D)}{2} \choose k}, {\frac {n+\Phi(D)}{d}+k-1
\choose k}\right\},$$ and $d$ is the smallest nontrivial divisor of
$|G|$ that is not equal to 2.
\end{cor}

\begin{cor}\label{cor3.6} Let $q\geq 64$ and $n=q+2$. For $6\leq k<q-1$, we have $N(k,b,D)>0$ for every $b\in G$.
\end{cor}
\begin{proof}
By symmetry it is sufficient to consider the case $3\leq k\leq n/2$.
To ensure $N(k, b, D)>0$, by (\ref{4.1}) it suffices to have
$${n \choose k}> |S|{\Phi(D)+k-1 \choose k}+{\frac {n+\Phi(D)}
2\choose k}+\sum_{2<d\leq k \atop d\mid \exp(G)}
 \phi(d){\frac {n+\Phi(D)}{d}+k-1 \choose k}.$$
  For a nontrivial character $\chi$, $\sum_{g\in G}\chi(g)=0$
 and it follows that  $\Phi(D)=\Phi(G-D)<|G|-|D|\leq 2\sqrt{q}+1$.

 Since $G$ is the
 product of at most two cyclic groups, by the
 definition of $\phi(d)$ we have
 $\phi(d)\leq d^2-1.$
For simplicity, set $K=k^3-2k^2-k+2$. For the case $k\leq q^{1/3}$,
it is sufficient to have
$${q+2 \choose k}-(q+2\sqrt{q}-K){2\sqrt{q}+k-1 \choose k}-{\frac {q+2+2\sqrt{q}}
2\choose k} -K{\frac {q+2\sqrt{q}}{3}+k-1 \choose k}>0.$$ When
$k=3$, one has
$$ 125/216q^3-379/36q^{5/2}-589/18q^2
+593/27q^{3/2}+149/2q+67/3q^{1/2}>0$$ It then suffices to have
$q>432$.

Similarly, when $k=6$, one has $q\geq 64$. This is done by first taking
$K=k^3-2k^2-k+2=140$, we solve that $q\geq 97$. But notice that now
$K$ should
 be $\leq 117$. Then taking $K=117$, we solve $q\ge 79$. Iteratively, we can get $q\geq 64$ finally.

One checks that when $k\leq q^{1/3}$ this function is unimodal on
$k$. For $q^{1/3}<k<(q+2\sqrt{q})/6$, it then suffices to have
$${q+2 \choose k}>(q+2+2\sqrt{q}){\frac{q+2+2\sqrt{q}}2 \choose k},
$$ and for $(q+2\sqrt{q})/6\leq k \leq (q+2)/2$,
$${q+2 \choose
k}>(q+2+2\sqrt{q}){\frac{q+2\sqrt{q}}3+k-1 \choose k}.
$$
It follows from a simple asymptotic analysis and the proof is
complete.
\end{proof}

 A similar argument gives
\begin{cor}\label{cor3.7}Suppose that $n\geq (\frac 23+\epsilon)q$ and
$q>\frac 4 {\epsilon^2}$,
where $\epsilon$ is positive. Then there is a positive constant
$C_{\epsilon}$ such that
 $N(k,b,D)>0$ for every $b\in G$ provided
 $C_{\epsilon}\ln{q}<k<n-C_{\epsilon}\ln{q}$.
\end{cor}

\begin{proof}
Similar to the proof of the corollary above, we consider the case $k\leq n/2$.
To ensure $N(k, b, D)>0$, by (\ref{4.1}) it suffices to have
$${n \choose k}> |S|{\Phi(D)+k-1 \choose k}+{\frac {n+\Phi(D)}
2\choose k}+\sum_{2<d\leq k \atop d\mid \exp(G)}
 \phi(d){\frac {n+\Phi(D)}{d}+k-1 \choose k}.$$
  For a nontrivial character $\chi$, $\sum_{g\in G}\chi(g)=0$
 and it follows that  $\Phi(D)=\Phi(G-D)<|G|-|D|\leq (\frac{1}{3}-\epsilon)q+2\sqrt{q}+1$.

 For small $k\leq q/6$
it suffices to have
$${\frac {2q}3 \choose k}-(q+2\sqrt{q}){\frac {q+2+2\sqrt{q}}
2\choose k}>0,$$
i.e.,
$$\frac {\left(2/3 q\right)_k} {((q+2+2\sqrt{q})/2)_k} \geq q+2\sqrt{q},$$
and then
$$\left(\frac {2/3 q} {(q+2+2\sqrt{q})/2}\right)^k \geq q+2\sqrt{q},$$
 which holds when
\[
   k>C\ln q
\]
for some constant $C$.

For $q/6<k\leq n/2=(\frac 13 +\frac \epsilon 2)q$, it suffices to have
$$
{(\frac 23+\epsilon)q \choose k}>(q+2\sqrt{q}){(\frac {2}{3}+\frac \epsilon 2)q+\sqrt{q}\choose k},
$$
which holds when $q>\frac{4}{\epsilon^2}$ and
\[
k>C_{\epsilon}\ln q
\]
for some constant $C_{\epsilon}$. So, the proof is complete.

\end{proof}

From the proof of the above corollary, if follows that
\begin{cor}\label{cor3.8}
Suppose $n\geq (\frac 23+\epsilon)q$, where $\epsilon$ is
positive and $\epsilon\leq 1/3$. When $q$ is large enough
(in application we need to use long length codes, so it is reasonable to assume $q$ is large),
then there is a positive constant $C$ (independent of $\epsilon$ and $q$) such that
 $N(k,b,D)>0$ for every $b\in G$ provided
 $C\ln{q}<k<n-C\ln{q}$.
\end{cor}

\section{Minimum Distance of Elliptic Codes and SSP}\label{sec4}
In this section, we discuss the relationship between the minimum distance of ECAG code and SSP
on the group of rational points of the elliptic curve. Using the results in the previous section,
our main theorems in Introduction follow automatically.

We fix some notations for this section:
 \begin{itemize}
  \item\emph{
 $X/\f$ is a geometrically irreducible smooth projective curve of genus $g$  over the finite field $\f$
with function field $\f(X)$. }
 \item \emph{ $X(\f)$ is the set of all $\f$-rational points on $X$.}
 \item\emph{ $D=\{P_{1},P_{2},\cdots,P_{n}\}$ is a proper
subset of rational points $X(\f)$.}
\item\emph{Without any confusion, we also write $D=P_{1}+P_{2}+\cdots+P_{n}$.}
\item\emph{ $G$ is a divisor of degree $k$
($2g-2<k<n$) with $\mathrm{Supp}(G)\cap D=\emptyset$.} \end{itemize}

Let $V$  be a divisor on $X$. Denote by $\mathscr{L}(V)$ the $\f$-vector space of all
rational functions $f\in \f(X)$ with the principal divisor
$\mathrm{div}(f)\geqslant -V$, together with the zero function
(cf.~\cite{Stichtenoth}). It is well-known that $\mathscr{L}(V)$ is
  finite dimensional vector space over $\f$ and $\dim \mathscr{L}(V)=k-g+1$.

The functional AG code $C_{\mathscr{L}}(D, G)$ is defined to be the image of the following evaluation map:
 \[
       ev: \mathscr{L}(G)\rightarrow \f^{n};\, f\mapsto
       (f(P_{1}),f(P_{2}),\cdots,f(P_{n})).\enspace
\]

As functions in $\mathscr{L}(G)$ have at most $\deg{G}$ zeros, the minimum distance of $C_{\mathscr{L}}(D, G)$ is $d\geqslant n-k$. Together with Riemann-Roch theorem, it is easy to see that the functional AG code $C_{\mathscr{L}}(D, G)$ has parameters $[n, k-g+1, d\geqslant n-k]$. By the Singleton bound, we have
\[
   n-k\leq d\leq n-k+g.\
\]

If $X=E$ is an elliptic curve over $\f$, we only have the following two choices
for the minimum distance of $C_{\mathscr{L}}(D, G)$:
\[
  d=n-k,\,\mbox{ or }\, d=n-k+1.\
\]

Let $\mathcal {G}$ be an abelian group with zero element $O$ and $D$
a finite subset of $\mathcal {G}$. For an integer $0<k<|D|$ and an
element $b\in D$, we denote
\[
  N_{\mathcal {G}}(k,b,D)=\#\{S\subseteq D\,|\,\#S=k\, \mbox{ and }\, \sum_{x\in S}x=b\}. \
\]
Computing $N_{\mathcal {G}}(k,b,D)$ is called the counting version of
the \emph{$k$-subset sum problem} ($k$-SSP). In general, the counting
$k$-SSP is \textbf{NP}-hard. If there is no confusion, we simply
denote
\[
    N(k,b,D)=N_{\mathcal {G}}(k,b,D).\
\]

Let $E$ be an elliptic curve defined over $\f$ with a rational
point $O$. The set of rational points $E(\f)$ forms an abelian group
with zero element $O$ (for the definition for the sum of any two
points, we refer to \cite{Silverman}), and it is isomorphic to the
Picard group
 $\mathrm{div}^o(E)/\mathrm{Prin}(\f(E))$ where $\mathrm{Prin}(\f(E))$
 is the subgroup consisting of all principal divisors. 
 Denote by $\oplus$ and $\ominus$ the additive and minus operator in the group $E(\f)$, respectively.
\begin{prop}[\cite{chengqi,ZFW}]\label{prop2.1}
Let $E$ be an elliptic curve over $\f$,
$D=\{P_{1},P_{2},\cdots,P_{n}\}$ a subset of $E(\f)$ such that
rational points (not necessarily distinct) $O,P\notin D$ and let $G=(k-1)O+P$ ($0<k<n$). Endow $E(\f)$ a group structure with the zero element $O$. Then the AG code $C_{\mathscr{L}}(D, G)$ is MDS, i.e., $d=n-k+1$ if and only if
\[
  N(k,P,D)= 0\ .
\]
And the minimum distance $d=n-k$ if and only if
\[
  N(k,P,D)>0\ .
\]

\end{prop}
\begin{proof}
We have already seen that the minimum distance of $C_{\mathscr{L}}(D, G)$ has two choices: $n-k$, $n-k+1$. So $C_{\mathscr{L}}(D, G)$ is not MDS, i.e., $d=n-k$ if and only if there is a function $f\in\mathscr{L}(G)$ such that the evaluation $ev(f)$ has weight $n-k$. This is equivalent to that $f$ has $k$ zeros in $D$, say $P_{i_1}, \cdots, P_{i_k}$. That is
\[
    \mathrm{div}(f)\geq -(k-1)O-P+(P_{i_1}+\cdots+P_{i_k}),
\]
which is equivalent to
\[
    \mathrm{div}(f)=-(k-1)O-P+(P_{i_1}+\cdots+P_{i_k}).
\]
The existence of such an $f$ is equivalent to saying
\[
  P_{i_1}\oplus\cdots\oplus P_{i_k}=P.
\]
Namely, $N(k,P,D)> 0.$
It follows that the AG code $C_{\mathscr{L}}(D, G)$ is MDS if and only if
$ N(k,P,D)=0.\ $
\end{proof}

\begin{rem}
In general, if $G$ is a divisor of degree $k$ on $E$, then for any rational point $Q\in E(\f)$, as $\deg(G-(k-1)Q)=1$, by Riemann-Roch theorem, there exists
one and only one rational point $P\in E(\f)$ such that $G\sim (k-1)Q+P$. Suppose there exist rational points $Q,P$ such that
$G\sim (k-1)Q+P$ and $P,Q\notin D$. Let $G'=(k-1)Q+P$.
Then the codes $C_{\mathscr{L}}(D, G)$ and $C_{\mathscr{L}}(D, G')$ are
equivalent \cite[Proposition~2.2.14]{Stichtenoth}. Here two codes $C_{1},C_{2}\subseteq \f^n$ are said to be
\emph{equivalent} if there is a vector $a=(a_{1},\cdots,a_{n})\in
(\f^*)^n$ such that
\[
  C_{2}=a\cdot C_{1}=\{(a_{1}c_{1},\cdots,a_{n}c_{n})\,|\,(c_{1},\cdots,c_{n})\in
  C_{1}\}\ .
\]
It is easy to see that two equivalent codes have the same weight distribution and hence the same minimum distance. So it suffices
to consider all AG codes of the form $C_{\mathscr{L}}(D, (k-1)Q+P)$.

\end{rem}



Proposition~\ref{prop2.1} establishes the relation between minimum distance of ECAG code and SSP on the rational point group of the elliptic curve. Together with Corollaries~\ref{cor3.6} and \ref{cor3.7}, we obtain the main results of this paper, Theorems~\ref{thm1.2}-\ref{thm1.5}.

{\bf Acknowledgements.}  This paper was written when the first
author was visiting the Department of Mathematics at the University of
Delaware and the Department of Mathematics at the University of
California, Irvine. The first author would like to thank both
departments for their hospitality.


\bibliographystyle{plain}
\bibliography{MDS}




\end{document}